\documentclass[12pt, openbib]{article}

\let\/\over 
\usepackage{graphicx, amsthm, amsmath, aliascnt, caption, microtype}

\usepackage[colorlinks]{hyperref}
\usepackage[norefs,nocites]{refcheck}
\makeatletter 
\newcommand{\refcheckize}[1]{
	\expandafter\let\csname @@\string#1\endcsname#1%
	\expandafter\DeclareRobustCommand\csname relax\string#1\endcsname[1]{%
	\csname @@\string#1\endcsname{##1}\wrtusdrf{##1}}%
	\expandafter\let\expandafter#1\csname relax\string#1\endcsname
}
\refcheckize\autoref

\newaliascnt{lemma}{theorem}
\newtheorem{lemma}[lemma]{Lemma}
\aliascntresetthe{lemma}

\newcommand\capitalize[1]{\char\the\uccode\expandafter`#1} 
\newcommand\Capitalize[1]{\edef#1{\capitalize#1}}
\Capitalize\sectionautorefname
\Capitalize\subsectionautorefname

\let\<\le
\let\>\ge 

\def\k_#1#2{k_\textit{#1#2}}
\def\p_#1#2{\Delta\phi_\textit{#1#2}}

\newcommand\xl{|x_L|}
\newcommand\yl{|y_L|}
\newcommand\yr{|y_R|}
\newcommand\xr{|x_R|}

\newcommand\nl{\\&\textstyle}

\title{A Linear Potential Function for Pairing Heaps}
\author{John Iacono \and Mark Yagnatinsky}
\begin{document}
\maketitle

\begin{abstract}
We present the first potential function for pairing heaps with linear range.  This implies that the runtime of a short sequence of operations is faster than previously known.  It is also simpler than the only other potential function known to give amortized constant amortized time for insertion.
\end{abstract}

\section{Introduction}
The pairing heap is a data structure for implementing priority queues introduced in the mid-1980s~\cite{pair}.  The inventors conjectured that pairing heaps and Fibonacci heaps~\cite{fib} have the same amortized complexity for all operations.  While this was eventually disproved~\cite{LBfred}, pairing heaps are, unlike Fibonacci heaps, both easy to implement and also fast in practice~\cite{bench}.  Thus, they are widely used, while their asymptotic performance is still not fully understood.  Here, we show that a short sequence of operations on a pairing heap takes less time than was previously known.

\subsection{A review of amortization and potential functions}
Some applications of data structures require that the worst-case time of operations be bounded, even at the expense of the average case.  For instance, a text editor that takes an average of a millisecond to respond to a user's keystrokes will not be popular if a keystroke occasionally takes a minute to process.  Far better that the average be increased to two milliseconds, if by doing so we can guarantee that no keystroke will take more than four.  (Note: any connection to reality that these numbers have is utterly accidental.)

But many applications of data structures are in batch algorithms, and in that context, the worst-case runtime of any particular operation on the data structure doesn't matter.  An algorithm performs not one operation, but a sequence of them, and it is the runtime of the entire sequence that matters.

Consider a sequence of operations on some data structure, and let $t_i$ denote the time to execute the $i$th operation in this sequence.  The goal of amortized analysis is to assign to each operation an \emph{amortized time} $a_i$, such that the sum of the amortized times is at most the sum of the actual times.  There is of course an easy way to achieve this, by setting $a_i=t_i$.  The real goal is to come up with times that are somehow more convenient, usually by showing an amortized time that is a function of the current size of the structure, but not its detailed internal state.

A key tool in amortized analysis is the \emph{potential function}.  Given a sequence of $m$ operations $o_1,o_2,\dots, o_m$ executed on a particular data structure, and an integer $i\ge0$, a potential function is a function $\Phi$ that maps $i$ to a real number $\Phi_i$.  The real number $\Phi_i$ is called the \emph{potential} after the $i$th operation.  Given $\Phi$, we define the amortized time $a_i$ of an operation $o_i$ as the actual time $t_i$ of the operation, plus the change in potential $\Phi_i-\Phi_{i-1}$.  Often, instead of explicitly providing the mapping from integers to reals, a potential function is specified implicitly: given the current state of the data structure, there is an algorithm to calculate the current value of the potential function, without needing to know how many operations have taken place.  Another common approach (which we use in combination with the previous one) is to specify $\Phi_0$, and then specify how to update $\Phi$ after an operation.

Observe that since $a_k = t_k + \Phi_k - \Phi_{k-1}$, we have $t_k = a_k - \Phi_k + \Phi_{k-1}$.  Thus, the total time taken for the subsequence of consecutive operations from $i$ to $j$ is
\begin{align*}\textstyle
\sum_{k=i}^j t_k&\textstyle
= \sum_{k=i}^j (a_k - \Phi_k + \Phi_{k-1})\nl
= \sum_{k=i}^j a_k - \sum_{k=i}^j \Phi_k + \sum_{k=i}^j \Phi_{k-1}\nl
= \sum_{k=i}^j a_k - \sum_{k=i}^j \Phi_k + (\Phi_{i-1} + \sum_{k=i+1}^j \Phi_{k-1})\nl
= \sum_{k=i}^j a_k - \sum_{k=i}^j \Phi_k + \Phi_{i-1} + \sum_{k=i}^{j-1} \Phi_k\nl
= \Phi_{i-1} - \Phi_j + \sum_{k=i}^j a_k.
\end{align*}
From this formula, we can derive the motivation for our result: namely, that a potential function with a small range is useful.  (When we speak of the range of a potential function, we mean its maximum value as function of $n$ minus its minimum value.  Often the minimum is zero, and thus the range is simply the maximum value.)  To see this, suppose you have a data structure with $n$ elements, and you perform a sequence of $k$ operations on it that don't change the size.  If each operation takes $O(\log n)$ amortized time, then the total actual time is bounded by $O(k\log n)$ plus the loss of potential.  Thus, if the range of the potential function is $O(n\log n)$, then the total time is $O(k\log n + n\log n)$, but if the range of the potential function is linear, this is improved to $O(k\log n + n)$, which is asymptotically better whenever $k$ is $o(n)$.  Thus, a reduced range of a potential function improves the time bounds for short sequences that don't start from an empty data structure.

\subsection{Heaps and priority queues}
A priority queue is an abstract data type that maintains a totally ordered set under the following operations:
\begin{description}
\item[$Q =$ make-pq$()$:] Create an empty priority queue $Q$.
\item[$p =$ insert$(Q,x)$:] Insert key $x$ into $Q$, and return a handle $p$ which can be passed to decrease-key().
\item[$x =$ get-min$(Q)$:] Return the minimum key currently stored in $Q$.
\item[delete-min$(Q)$:] Delete the minimum key currently stored in $Q$.
\item[decrease-key$(p, y)$:] Decrease the key at $p$ to $y$.  Precondition: $y$ must be strictly less than the original key.
\end{description}
Priority queues come up in such a wide variety of applications that they are covered in (nearly?) all data structure textbooks.  As already mentioned, Dijkstra's shortest path algorithm (and its close relative, Prim's algorithm for minimum spanning trees) relies on a good priority queue for efficiency on sparse graphs.  As another example, the heapsort sorting algorithm (insert all items into a heap, then remove them one by one) is an efficient in-place sorting algorithm given the right priority queue.  (Specifically, when using a binary heap.)

The simplest implementation of a priority queue is probably a singly linked list.  Initializing a new heap is trivial, and inserting a new item is simply an append, which takes constant time.  If we return a pointer to the node an item is stored in, we have a suitable handle to perform decrease-key in constant time.  Unfortunately, finding the minimum item takes linear time.  Likewise, deleting the minimum item also takes linear time, since we must first find it.  If we modify the insert operation to maintain the current minimum, then the current minimum can be found in constant time, but deletion still takes linear time, because we then have to find the \emph{new} minimum.  For this reason, this implementation is rarely used if the priority queue has any chance of being the algorithmic bottleneck.

Another approach is to use a balanced binary search tree.  Its advantage over using a linked list is that deletion of the minimum now takes only logarithmic time, instead of linear.  The disadvantage is that insertion slows down to logarithmic, as does decrease-key.  Furthermore, the work required to implement a binary search tree noticeably exceeds that needed for a linked list or array.

The problem is that a binary search tree is trying too hard: it must be ready to support finding an arbitrary element at all times, despite the fact that we will only ever ask it for the smallest.  Meanwhile, the unordered list has the opposite problem of trying too little.  It turns out that instead of a search tree, it helps to store the set of items in a \emph{heap}.  There are many varieties of heaps, but they all have the following property in common.  Like binary search trees, they store the elements of the set in nodes of a rooted tree (or sometimes a forest of such trees).  Unlike search trees, heaps have the invariant that all children of a node have larger key values than their parent.

\subsection{Pairing heaps}
A \emph{pairing heap}~\cite{pair} is a heap-ordered general rooted ordered tree.  That is, each node has zero or more children, which are listed from left to right, and a child's key value is always larger than its parent's.  The basic operation on a pairing heap is the \emph{pairing} operation, which combines two pairing heaps into one by attaching the root with the larger key value to the other root as its leftmost child.  For the purposes of implementation, pairing heaps are stored as a binary tree using the leftmost-child, right-sibling correspondence.  That is, a node's left child in the binary tree corresponds to its leftmost child in the general tree, and its right child  in the binary tree corresponds to its right sibling in the general tree.  In order to support decrease-key, there is also a parent pointer which points to the node's parent in the binary representation.  Priority queue operations are implemented in a pairing heap as follows:
\begin{description}
\item[make-heap():] return null
\item[get-min($H$):] return $H$.val
\item[insert($H, x$):] create new node $n$ containing $x$; If the root $H$ is null, then $n$ becomes the new root; if $H$ is not null then pair $n$ with $H$ and update root; return pointer $p$ to the newly created node.
\item[decrease-key($p, y$):] Let $n$ be the node $p$ points to.  Set the value of $n$'s key to $y$, and if $n$ is not the root, detach $n$ from its parent and pair it with the root
\item[delete-min($H$):] remove the root, and then pair the remaining trees in the resultant forest in groups of two.  Then incrementally pair the remaining trees from right to left.  Finally, return the new root.  See \autoref{fig:delete} for an example of a delete-min executing on a pairing heap.  (Readers familiar with splay trees may notice that in the binary view, a delete-min resembles a splay operation.)
\end{description}
All pairing heap operations take constant actual time, except delete-min, which takes time linear in the number of children of the root.  Pairing heaps naturally support one more operation in constant time: merge.  This takes two independent heaps and pairs them.  Unfortunately, this takes amortized linear time using our potential function.

\begin{figure}
\def\fig#1#2{\hfil\includegraphics{del#1}\hfil

#2}
\def\w{\dimexpr.5\textwidth-2\tabcolsep}
\def\nl{\\\hline\\[-.8\normalbaselineskip]}
\begin{tabular}{p{\w}|p{\w}}
\fig1{\hfil (a) Remove the root.}
&
\fig2{(b) The first pass groups the nodes in pairs, and pairs them.}
\nl
\fig3{(c) The second pass repeatedly pairs the right two nodes until a single tree is formed.}
&
\fig4{\hfil (d) Second pairing pass, continued.}
\nl
\fig5{\hfil (e) Second pairing pass, continued.}
&
\fig6{\hfil (f) Final heap after a delete-min.}
\end{tabular}
\caption{Delete-min on a heap where the root has eight children.}\label{fig:delete}
\end{figure}

\subsubsection{History}
Pairing heaps were originally inspired by splay trees~\cite{splay}.  Like splay trees, they are a self-adjusting data structure: the nodes of the heap don't store any information aside from the key value and whatever pointers are needed to traverse the structure.  This is in contrast to, say, Fibonacci heaps~\cite{fib}, which store at each node an approximation of that node's subtree size.  Fibonacci heaps support delete-min in logarithmic amortized time, and all the other heap operations in constant amortized time.  However, they are complicated to implement, somewhat bulky, and therefore slow in practice~\cite{exp}.  Pairing heaps were introduced as a simpler alternative to Fibonacci heaps, and it was conjectured that they have the same amortized complexity for all operations, although~\cite{pair} showed only an amortized logarithmic bound for insert, decrease-key, and delete-min.  The conjecture was eventually disproved when it was shown that if insert and delete-min both take $O(\log n)$ amortized time, an adversary can force decrease-key to take $\Omega(\log\log n)$ amortized time~\cite{LBfred}.

\subsubsection{Present}
Nevertheless, pairing heaps are fast in practice.  For instance, the authors of~\cite{bench} benchmarked a variety of priority queue data structures.  They also tried to estimate difficulty of implementation, by counting lines of code, and pairing heaps were essentially tied for first place by that metric, losing to binary heaps by only two lines.  Despite (or rather because of) their simplicity, pairing heaps had the best performance among over a dozen heap variants in two out of the six benchmarks.  In one of the benchmarks in which pairing heaps did not come in first, they were within ten percent of the performance of the heap which did, and in two others, they were within a factor of two of the best.  The one benchmark in which they did poorly was pure sorting (add everything to the heap and then remove it), where they were over four times slower than the fastest heap.  (Although it might be worth pointing out that the heap that won that benchmark does not support decrease-key in sub-linear time at all, so the comparison is not quite apples-to-apples, especially since it is possible to save one pointer per node in a pairing heap if you know you won't perform decrease-key.  Pairing heaps were only thirty percent slower in the sorting benchmark than the fastest heaps that did support decrease-key.)

\subsection{Previous work and our result}
In~\cite[Theorem~3]{cole}, Cole develops a linear potential function for splay trees (that is, the potential function ranges from zero to $O(n)$), improving on the potential function used in the original analysis of splay trees, which had a range of $O(n\log n)$~\cite{splay}.  As explained above, this allows applying amortized analysis over shorter operation sequences.

There are several variants of pairing heaps such as~\cite{sort} and~\cite{elm}, and one of them also has a potential function that is $o(n\log n)$~\cite{sort}.  The main theme in all the variants is to create a heap with provably fast decrease-key, while maintaining as much of the simplicity of pairing heaps as possible.

\paragraph{Our result.}  We present a potential function for pairing heaps that is much simpler than the one found for splay trees in~\cite{cole} and also simpler than the only previously known potential function for pairing heaps that is $o(n\log n)$~\cite{pet}.  Further, it is simpler than the only other potential function known to give constant amortized time for insertion~\cite{o1}, and perhaps more importantly, it is the first potential function for pairing heaps whose range is $O(n)$, which allows the use of amortized analysis to bound the run times of shorter operation sequences than before.  In the case of pairing heaps, this bound on the potential function range is asymptotically the best possible, since the worst-case time for delete-min is linear, and thus we need to store at least a linear potential to pay for it.

\begin{table}
\newcommand\pet{4^{\sqrt{\lg\lg n}}}
\newcommand\SV{Stasko/Vitter }
\newcommand\ELM{Elmasry }
\footnotesize
\centerline
{\begin{tabular}{@{} ll ll l@{}}
Result					&	Range				&	Insert		&	Decrease-key	&	Delete-min\\
\hline
Pairing heap \cite{pair}&	$\Theta(n\lg n)$	&	$O(\lg n)$	&	$O(\lg n)$&		$O(\lg n)$\\
Pairing heap \cite{pet}&	$O(n\cdot\pet)$	&	$O(\pet)$	&	$O(\pet)$&			$O(\lg n)$\\
Pairing heap \cite{o1}&	$O(n\lg n)$		&	$O(1)$		&	$O(\lg n)$&		$O(\lg n)$\\		
Pairing heap [This paper]&$\Theta(n)$		&	$O(1)$		&	$O(\lg n)$&		$O(\lg n)$\\
\hline
\SV\cite{exp} &			$O(n\lg n)$		&	$O(1)$		&	$O(\lg n)$&		$O(\lg n)$\\
\ELM\cite{elm,elm2}&		$O(n\lg n)$		&	$O(1)$		&	$O(\lg\lg n)$&		$O(\lg n)$\\
Sort heap \cite{sort}&		$\Theta(n\lg\lg n)$&	$O(\lg\lg n)$&	$O(\lg\lg n)$&		$O(\lg n\lg\lg n)$\\	
\hline
Binomial heap \cite{bin}&	$\Theta(\lg n)$&		$O(1)$		&	$O(\lg n)$&		$O(\lg n)$\\
Fibonacci heap \cite{fib}&	$\Theta(n)$ &			$O(1)$		&	$O(1)$&			$O(\lg n)$\\
Rank-pairing heap \cite{rank}&$\Omega(n)$&	$O(1)$		&	$O(1)$&			$O(\lg n)$
\end{tabular}}
\caption{Various heaps and amortized bounds on their running times.  Top: analyses of pairing heaps.  Middle: close relatives of pairing heaps.  Bottom: more distant relatives of pairing heaps.  Note: $\lg=\log_2$.}
\label{table}
\end{table}

\paragraph{Previous work.}  In \autoref{table}, we list the amortized operation costs and ranges of several potential functions.  Each row of the table corresponds to a single analysis of a specific heap variant.  The table is divided into three parts.  The top part is devoted to analyses of pairing heaps.  The middle is for variants of pairing heaps, and the bottom is for heaps that are sufficiently different from pairing heaps that calling them a variant seems inaccurate.  It is of course somewhat subjective whether a result belongs in the second or third group.  For instance, a case could be made that rank-pairing heaps could go either way.  Also, note that the use of $O()$ notation in the Range column is deliberate.  First, because the hidden constants in the various results differ, sometimes dramatically.  Second, because most papers do not state the range of the potential function (with~\cite{sort} being a notable exception), using $\Theta()$ would force us to prove matching upper and lower bounds for each heap in the table, which would blow up the size of this subsection, which would then probably more than double the length of this paper.  There are several cases in which we did use $\Theta$; this is not meant to imply that the potential function is always this large, only that there exists at least one family of heaps for which it is that large.  For instance, the classic potential function for pairing heaps is $\Theta(n\log n)$ because inserting $n$ items in sorted order (increasing or decreasing both work) creates a heap with a potential that large, although some other operation sequences do not.  Below we justify the correctness of the Range column of the above table; for the other columns, see the respective papers.

We use $\lg$ for $\log_2$ hereafter.

\subparagraph{Classic pairing heap potential \cite{pair}:}  The upper bound here is trivial, since the heap potential is the sum of the node potentials, and the potential of a node is the binary logarithm of the size of its subtree in the binary view.  To get the lower bound, consider insertion of a sorted sequence (increasing or decreasing).  Then in the binary view, all nodes (except the one leaf) have exactly one child.  Thus, at least half of all nodes have a subtree size of $n/2$, and thus a potential of $\lg\frac n2 = \lg n - 1$, which makes the heap potential at least $\frac n2 (\lg n - 1) = \frac12 n\lg n - n/2$, which is $\Omega(n\log n)$.

\subparagraph{Seth Pettie's analysis \cite{pet}:}  Here the potential function is, to say the least, not very simple.  As a warm-up, the paper starts with a simpler potential function which is far simpler to analyze. (Even the simple one has some interesting features, such as depending on the current state of the heap, and on its past, and on the future.) Deriving bounds for the general one is left as an exercise for the reader.  (A grueling exercise, at that.)  As for the simpler one, a node's potential is clearly upper-bounded by $2\sqrt{\lg n}$, so it follows that the range is at most $n$ times that.  (Actually this does not follow immediately, since the potential function has another term which we're neglecting, but this term is too small to matter in this case.)  As for the lower bound, consider inserting $n$ items into the heap in increasing order.  (Decreasing order doesn't work here, because the potential function is based on the general view of the heap instead of the binary view, so it is less symmetric.)  Then at least half of all nodes have $\hbar=\lg\frac n2 = \lg n - 1$ (see the paper for the definition of $\hbar$), and thus a potential of $\Omega(\sqrt{\lg n})$.  Therefore, the total heap potential would be $\Theta(n\sqrt{\lg n})$.  Conjecture: the same construction works for the fancy potential function, which would give a bound of $\Theta(n\cdot4^{\sqrt{\lg\lg n}})$.

\subparagraph{John Iacono's analysis \cite{o1}:}  The upper bound of $O(n\log n)$ follows immediately from the definition of the potential function (which we omit).  To show that this bound is tight, some background will be needed.  This potential function depends on the future even more strongly than that of Seth Pettie.  For this potential function, it makes little sense to speak of the amortized run time of a single operation, even in a sloppy, informal sense, unless one has in mind, at least implicitly, the operation sequence that this operation is part of.  Thus, if we wish to speak of the range of this potential function, it makes sense to look at its value at the beginning, and at the end, and then subtract the two.  All intermediate values are essentially internal to the analysis.  And at the end of an operation sequence, the potential is always at least zero and at most linear.  However, the potential function depends only on the future and not the past, so it is legitimate to start the operation sequence from a nonempty heap.  If we start it from the heap that results from inserting $n$ items in decreasing order, and let the operation sequence be ``delete-min, $n$ times,'' then the initial potential equals $-\Theta(n\log n)$ (yes, negative!), and thus the bound is tight.  Note, however, that this is not as interesting as it seems.  Recall that the point of having a small range is to gain the ability to amortize over short sequences of operations.  But we have created a sequence of $n$ delete-min operations, which is not short, since such a sequence naturally takes $\Theta(n\log n)$ time anyway, so the large range does no harm.  Thus the tightness of the upper bound remains mysterious.

\subparagraph{Stasko and Vitter's variant(s) \cite{exp}:}  This paper actually defines two variants of pairing heaps.  They both use the same potential function, and the upper bound follows easily since the potential function is nearly identical to the classic potential of~\cite{pair}.  The variants differ in how decrease-key is implemented.  The first variant uses a very simple implementation of decrease-key, similar in spirit to classic pairing heaps.  Unfortunately, the authors were not able to analyze this variant, beyond the fact that insertion takes constant time, and deletion takes logarithmic amortized time, if no decrease-key operation is ever performed.  The second variant is mentioned only in passing: it uses a relatively heavy-weight implementation of decrease-key.  In this case, they were able to show that decrease-key takes at most logarithmic amortized time (although this result was merely sketched).  Both variants are robust against the sorts of tricks we used to establish the lower bounds of the preceding potential functions.  Conjecture: there exists a sequence of insert, decrease-key, and delete-min operations that raise the potential of this heap to something more than linear, and most likely all the way to the upper bound of $O(n\log n)$.

\subparagraph{Amr Elmasry's variants:}  The potential in~\cite{elm} is $O(n\log n)$ by inspection.  It appears to be even harder to prove this bound is tight than in Stasko and Vitter's variant, for similar reasons.

In the case of~\cite{elm2}, an upper bound of $O(n\log n)$ is also easy to see by inspection.  However obtaining a good lower bound is harder.  The amortization argument is somewhat complex, using a combination of potential, credits, and debits.  The credits are mostly just terms in the potential function by another name: some operations create credits and later operations use them.  The potential functions, as written, would cause insertion into the pairing heap to take more than constant time.  (And likewise for melding.)  The purpose of the debits is to get around this problem, by making later deletions subsidize the earlier insertions.  Rephrasing this in terms of potential functions is not as natural as in the case of credits, but should be doable.  Once this is done, it is by no means clear how to construct a heap with large~potential.\looseness-1

\subparagraph{John Iacono's variant: sort heaps \cite{sort}.}  Here the upper and lower bounds are again easy.  The heap potential is the sum of the node potentials, and the potential of a node is at most $\lg\lg n$.  This is tight: insert $n$ nodes in increasing order, and one third of them will be right heavy, and thus have a potential of $\lg\lg n$.

\subparagraph{Binomial Heaps \cite{bin}:}  This heap barely qualifies for inclusion into the table, since it performs all operations in worst-case logarithmic time, and we are only interested in amortized heaps.  But, in fact, the amortized time of insertion is $O(1)$, and the proof is so simple that we will sketch it here.  A binomial heap is an ordered forest of rooted trees, each satisfying the heap-order property (that each child's key exceeds that of its parent).  The size of each tree is a power of two, and no two trees have the same size. (The forest is sorted by tree size.) This set of constraints is already enough to allow us to deduce the structure of the heap given only its size, $n$: simply write $n$ in binary and if, say, the eights place has a 1, then the binomial heap will have a tree of size 8, otherwise it won't.  If, in the course of performing an operation, the unique size condition is violated, the two trees with the same size are merged into one.  (This is analogous to the pairing operation in pairing heaps, but in the context of binomial heaps is called merging instead.)  Thus, insertion is performed by adding a tree of size 1, and then performing as many merges as needed to restore the unique size condition.  If we let the potential be the number of trees in the forest, then it is trivial to show that insertion takes $O(1)$ amortized time: any merge takes one unit of work and releases one unit of potential, so they are effectively free.  The range of the potential function is clearly logarithmic, since there are no more than $\lg(n+1)$ trees in the heap (and usually less, unless $n$ is one less than a power of two).

\subparagraph{Fibonacci heaps: \cite{fib}.}  A delete-min in a Fibonacci heap may take linear time, so the potential range must be at least linear or else the potential function can't even be used to prove that a Fibonacci heap can sort $n$ items in $O(n\log n)$ time.  The potential function for Fibonacci heaps is almost as simple as that for binomial heaps: the number of trees in the forest, plus twice the number of marked nodes.  (The potential for marked nodes is only important for proving that decrease-key takes constant amortized time.)  Clearly the range is never more than linear, since one can't do worse than marking every node and placing it in its own tree.

\subparagraph{Rank-pairing heaps: \cite{rank}.}  Again, the potential must be at least linear.  The paper analyzes two variants of rank-pairing heaps, and in both cases the analysis is involved.  Here we make a few superficial observations.  The potential function depends only on the current state of the structure, rather than knowing the past or the future.  Of the three primary operations we are concerned with, delete-min in general releases potential.  Likewise, decrease-key in general needs to release potential, since in this heap, as in Fibonacci heaps, the worst-case time is linear.  Thus, the only operation that allows us to build up potential is insertion, and its amortized cost is constant.  Thus we should expect a typical rank-pairing heap to have linear potential.  Unfortunately, in the worst case, this kind of argument proves nothing.  For instance, even though a typical decrease-key should release potential, there could be a sequence of cleverly chosen decrease-key operations such that each of them costs a unit of potential (but no more than that, decrease-key takes constant amortized time).  Given a long enough such sequence, the heap potential would eventually rise to more than linear.  Proving such a sequence does not exist requires looking at the potential function in detail.  Without that, at best we can say it seems implausible, since such a sequence will be applying decrease-key to the same nodes repeatedly, which is not going to change the heap structure, since they are already roots, unless some delete-min operations are mixed in.  Neither can we rule out a clever sequence of delete-min operations, each raising the potential by $O(\log n)$; such a sequence can't be very long, since eventually the heap will be empty and have a potential of zero.  However, we again recall that the reason we care about this at all is that we wish to amortize over short sequences.  Such tricks, even if they were to work, are nearly impossible to pull off with a short sequence.  (But not quite impossible: imagine a sequence of $n/\sqrt{\lg n}$ deletions, each raising the potential by $\lg n$.  At the end, we would have a heap of size $n-o(n)$ with potential $n\sqrt{\lg n}$.)

\paragraph{The analysis of pairing heaps presented in this paper:}  In \autoref{linear} we present a new potential function for pairing heaps and prove that its range is linear.  This is optimal, because we must release linear potential during a delete-min in the worst case, or else the amortized time of delete-min would equal its actual time, which is linear, whereas a good heap should have logarithmic amortized time for delete-min.

\section{Towards a functional potential function}
Our goal is to construct a potential function which is always between zero and $O(n)$, that gives us logarithmic amortized time for delete-min, where $n$ is the current size of the heap.  Since the classic potential function already achieves the second goal, it is natural to attempt to tweak it somehow to achieve the first one without breaking anything.  The classic potential function assigns each node $x$ a potential of
$\lg|x|$, where $|x|$ is the size of the subtree rooted at $x$ in the binary view.  Suppose we simply try to scale down: the classic potential is too large by a logarithmic factor, so we assign each node a potential of
$\lg|x|\/\lg n$ instead.  Unfortunately, this means that a delete-min releases only $O(n/\log n)$ potential, leading to a linear amortized cost.  Another idea, if we don't mind randomization, is to flip a biased coin each time a node is inserted into the heap, and, if it lands on tails, then that node will have the classic potential, and otherwise none at all.  But if the coin is biased enough to get the potential low enough, then most nodes in the heap will have no potential at all and the amortized cost of delete-min will again be too high.

Let's try something else.  Some nodes have very small subtrees, and thus don't contribute much to the potential anyway.  Suppose we assign a potential of zero to a node if its subtree size is less than $\lg n$.  Pairing two such small nodes certainly won't release any potential, but neither is it likely to cost any potential.  Thus, if we have to pair $2s$ small nodes during a delete-min, we will have to pay for those pairings out-of-pocket; that is, they will contribute a cost of roughly $s$ to the cost of delete-min.  However, the root of the heap can't possibly have more than $\lg n$ small nodes as children (in the general view), since if a node is small, it must have less than $\lg n$ right siblings.  So, our next attempt is to use the classic potential for nodes that are not small, and zero potential for those that are.  Then we get the desired run time for delete-min, but the potential can be too large.  In particular, inserting nodes in sorted order (reverse sorted also works) yields a heap that looks like a path.  Most nodes will use the classic potential, and thus the total potential will be $\Omega(n\log n)$.  However, we are making progress.  If we could get the binary view of the heap into the shape of a balanced binary tree, the potential would indeed be linear.  But instead we ended up with a long path.

Suppose we insisted that both the left and the right subtrees of a node must have size at least $\lg n$, and otherwise the node gets no potential.  This certainly solves the problem with long paths having too much potential, but it goes too far: now long paths have no potential at all, and we may well encounter them when performing a delete-min.  We need something in between.  We end up with three types of nodes.  Small nodes where both the left and right subtrees have sub-logarithmic size have no potential.  Large nodes where both subtrees have at least logarithmic size have the classic potential.  And finally mixed nodes where exactly one of the subtrees has logarithmic size smoothly interpolate between the two extremes.  The potential of a mixed node is simply the classic potential times $s/\lg n$, where $s$ is the size of its smaller subtree.  As we will show soon (see \autoref{linear}), this is enough to guarantee a linear range for the potential function.  Unfortunately, this does not guarantee that delete-min takes amortized logarithmic time.

There is an annoying issue we've been sweeping under the rug ever since we introduced the idea of thresholds where the potential of a node changes sharply.  The classification of many nodes could change simultaneously simply because an insertion into the heap increased $n$, and thus the cutoff $\lg n$, or because a deletion decreased it, and this could have a large effect on the total potential of the heap.  We solve this problem in a similar spirit to how one implements a dynamic array.  A dynamic array has a capacity which always moves by large jumps; for instance one can constrain it to be a power of two.  Changing to a new capacity takes much time, but there is always enough potential saved up to pay for it.  Our analogue of capacity is the sticky size $N$.  Our cutoff will be based on $\lg N$, rather than $\lg n$.  The sticky size starts at 1, and is doubled after an insertion if $n$ is at least twice $N$, and is halved after a deletion if $N$ is at least twice $n$.  This allows us to introduce a size potential for the heap, equal to $9|N-n|$.  We will see that this is more than enough to allow us to pay for the changes in potential caused by changed node classifications due to changed $N$.

Our potential function almost works now, and we are ready to consider a proof strategy for delete-min.  As in the classic analysis, we analyze the first pairing pass and the second pairing pass separately.  For the second pairing pass, we will be content to show that it causes at most a logarithmic increase in potential.  Then we will show that the first pass releases enough potential to pay for all pairings performed in both passes.  (Since the actual work done by both passes is the same, we can certainly pay for both passes if we can pay for at least one, simply by doubling the potential of the heap; that is, instead of assigning a potential of $\lg n$ to each node, assign a potential of $2\lg n$.  We will return to the issue of constant factors shortly.)  Thus we simply show that a pairing between two large nodes (in the first pass) must usually release a unit of potential (with at most a logarithmic number of exceptions), and likewise for a pairing between two mixed nodes, and a pairing between a mixed and a large node, and so on.

We should expect a pairing between two large nodes to give us little trouble, since we should be able to reuse the classic analysis.  We can also expect little trouble if at least one of the nodes involved in the pairing is small, since there are so few small nodes involved in a delete-min operation.  Thus, we should look for trouble in the cases mixed-mixed and mixed-large.  It turns out there is a sub-case of a first-pass mixed-mixed pairing using this potential function where the potential of the heap ends up increasing by $O(1)$, instead of decreasing.  In that case just one of the nodes involved becomes small, which opens up the following simple patch to our potential function: small nodes still have a potential of zero, while mixed and large nodes have a potential of $4+\phi_x$, where $\phi_x$ is the potential of the node $x$ under the slightly broken potential function we are trying to fix.  The constant 4 is somewhat arbitrary, but the idea is that we offset the $O(1)$ potential increase by a $O(1)$ decrease caused by a node going from mixed to small.

Finally, we turn to the mixed-large pairings of the first pass.  It turns out that with the potential function we have so far, mixed-large pairings fail to release any potential.  However, most such pairings performed during the first pass result in the winner of the pairing being large.  (That is, the node with the smaller key value will be parent of the node with the larger key value, and this node with a small key will be a large node at the completion of the pairing.)  This does not seem helpful at first, until we realize that if three consecutive mixed-large pairings are performed, this means we have a chain of six nodes, three of which are large.  At most two of those large nodes were adjacent before the three pairings, and afterwards all three are consecutive siblings.  Thus, if we assign potential to edges (in the binary view) as well as to nodes, we have our final solution.  Most edges will have zero potential, except for those that connect two large siblings, which have a slight negative potential.  (More precisely, an edge has negative potential if it connects a large parent to its large left child in the binary view.)  We then split our analysis into two cases: either there are few mixed-large pairings, in which case they are subsidized by the potential released by the other pairings, or else there are many mixed-large pairings: so many that we can count on many such consecutive triples.  It remains to adjust the constants so that everything works out.  The potential on the negative edges must be $-3$ or less, since a single edge may have to pay for three consecutive mixed-large pairings.  This is doubled to $-6$ when we recall that the first pass has to pay for the second pass.  And finally, since we must also pay for the mixed-large pairings that didn't occur as part of a triple, we need a bit extra.  In the interest of sticking with small integers, we set the edge potential to $-7$.  We now have to also increase the potential of the mixed and large nodes so that if most pairings are not mixed-large, we can pay for the few that are.  In the interest of working with nice round numbers, we achieve this by multiplying the potential by one hundred.

\section{The potential function}
Our potential function is the sum of three components.  The first is the node potential, which will give a value to each node.  (The total node potential is the sum of the values for individual nodes.)  The second is the edge potential: each edge will have a potential of either 0 or $-7$.  (The total edge potential is likewise the sum of the values of individual edges.)  The third we shall call the size potential.  We begin with explaining the concept of the sticky size, since we will need it to define all three components.  The size $n$ of a heap is how many elements it currently stores.  The sticky size $N$ is initially 1.  After every heap update, the sticky size is updated as follows: if $n\>2N$, then $N$ is doubled, and if $n\<N/2$, then $N$ is halved.  The sticky size is the only aspect of the potential function which is not computable simply from the current state of the heap but is based on the history of operations.

The size potential is simply $900|N-n|$.

The node potential is slightly more complicated.  Given a node $x$, let $|x|$ be the size of its subtree (in the binary view).  Its left child in the binary view is $x_L$ and its right child is $x_R$.  The node potential $\phi_x$ of $x$ depends on $|x_L|$ and $|x_R|$.  Note that $|x| = |x_L|+|x_R|+1$.  Let $\lg=\log_2$.  There are three cases:
\begin{description}
\item[Large node.] If $\xl>\lg N$ and $\xr > \lg N$, then $x$ is \emph{large} and $\phi_x = 400+100\lg|x|$.
\item[Mixed node.] If $\xl\<\lg N<\xr$, then $x$ is \emph{mixed} and $\phi_x = 400+100\frac{\xl}{\lg N}\lg|x|$.  The case where $\xr\<\lg N<\xl$ is symmetric.
\item[Small node.] If $\xl \< \lg N$ and $\xr \< \lg N$, then $x$ is \emph{small} and $\phi_x=0$.
\end{description}
If the right child of a large node is also large, then the edge potential of the edge connecting them is $-7$.  All other edges have zero edge potential.

We define the actual cost of an operation to be one plus the number of pairings performed.  Through most of the analysis, it will seem that the node potential is a hundred times larger than what is needed.  Near the end, we will see that we use the excess to pay for mixed-large pairings during a delete-min.

\section{The analysis}
In the proofs below, we assume for convenience that the heap has at least four elements, so we can say things like ``the root of the heap is always mixed, since it has no siblings and $n-1$ descendants in the general representation,''  which assumes that $n-1>\lg N$, which may not be true for heaps with less than four elements.  If we need stronger assumptions on the size, we will call those out explicitly.

\begin{lemma}  The potential of a pairing heap is $O(N) = O(n)$. \label{linear}  \end{lemma}
\begin{proof}  The size potential is linear by definition.  The edge potential is slightly negative: most edges have potential zero, the exception being those edges that connect two large nodes.  If there are $L$ large nodes, the total edge potential may be as low as $-7(L-1)$.  But the node potential is at least $400L$, so the edge potential can never make the heap potential negative. 

We now turn to the node potential.  The small nodes have potential zero.  Observe that the lowest common ancestor (in the binary representation) of two large nodes must itself be large.  This immediately implies that if the left subtree (in the binary view) of a large node contains any large nodes, then this subtree contains a unique large node which is the common ancestor of all large nodes in this subtree.  (And likewise for the right subtree.)  Thus, it makes sense to speak of the tree induced by the large nodes, which is the binary tree obtained by taking the original pairing heap and performing the following two operations: first, erase all small nodes (they have no mixed or large descendants). After this, all mixed nodes have only one child, so they form paths.  Second, contract these paths to edges.  Now only large nodes remain, and ancestor-descendant relations are preserved.  (One thing that is not preserved is the root: in the original tree, the root is a mixed node.)  If a large node has less than two children in this shrunken tree, call it a \emph{leaf-ish} large node, and otherwise call it an \emph{internal} large node.  In the original tree, a leaf-ish large node $x$ has a left and a right subtree (in the binary representation), at least one of which has no large nodes in it: call that a \emph{barren subtree} of $x$; see \autoref{fig:shrink}.

\begin{figure}\centering\includegraphics{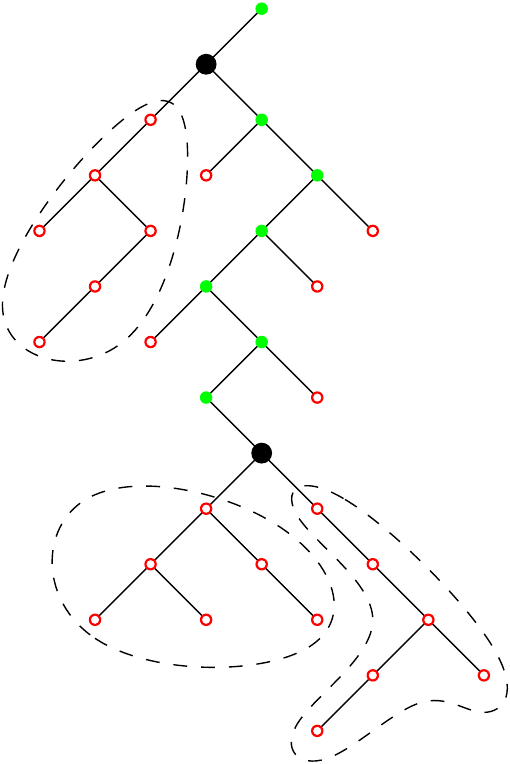}
\caption[How to make a shrunken tree]{The barren subtrees are marked by dashed bubbles; small nodes are marked with hollow red circles, large nodes are marked with bold black disks, and mixed nodes are green.}\label{fig:shrink}
\end{figure}

Let $x$ and $y$ be two leaf-ish large nodes, and observe that their barren subtrees are disjoint, even if $x$ is an ancestor of $y$.  We say that a leaf-ish large node \emph{owns} the nodes in its barren subtree.  Observe that, in the binary view, every leaf-ish large node has at least $\lg N$ descendants which are not owned by any other leaf-ish large node, which implies that there are at most $n/\lg N$ leaf-ish large nodes.  There are more leaf-ish large nodes than there are internal large nodes (since a binary tree always has more leaf-ish nodes than internal nodes, as can be shown by induction), so there are at most $2n/\lg N$ large nodes in total, each of which has a potential of at most $100\lg n+400=100\lg N+O(1)$, so the total potential of all large nodes is at most $200n + o(n)$.

This leaves the mixed nodes.  Every mixed node has a \emph{heavy subtree} with more than $\lg N$ nodes in the binary view, and a possibly empty \emph{light subtree} with at most $\lg N$ nodes.  Since the light subtree contains no mixed nodes, every node in the heap is in the light subtree of at most one mixed node; that is, the light subtrees of different nodes are disjoint.  If $x$ is a mixed node and $L_x$ is the size of its light subtree, the node potential of $x$ is
\begin{align*}\phi_x&\textstyle
= 400+100\frac{L_x}{\lg N}\lg|x|\nl
\<400+100\frac{L_x}{\lg N}\lg n\nl
\<400+100\frac{L_x}{\lg N}\lg2N\nl
= 400+100\frac{L_x}{\lg N}(1+\lg N)\nl
= 400+100\frac{L_x}{\lg N}+100L_x\nl
\<400+100\frac{L_x}{\lg4}+100L_x\nl
= 400+150L_x.
\end{align*}

Summing the potential over all mixed nodes, we have $\sum_x(400 + 150L_x) = 400n + 150 \sum_x L_x$.  Since all the light subtrees are disjoint, $\sum_x L_x$ is at most the heap size: $n$.  Thus, the combined potential of all mixed nodes is $550n$, and that of all nodes is $750n+o(n)$.
\end{proof}

We now analyze the five heap operations.  All operations except delete-min take constant actual time.  The get-min operation does not change the heap so its amortized time is also obviously constant.  The make-heap operation creates a new heap with a potential of 900, due to the size potential.  This leaves insertion, decrease-key, and deletion, which we handle in that order.

\begin{lemma} Insertion into a pairing heap takes $O(1)$ amortized time. \end{lemma}
\begin{proof} The actual time is constant, so it suffices to bound the change in potential.  We first bound the potential assuming $N$ stays constant during the execution of the operation.  Note that we need not worry about edge potentials here, since inserting a new node can not disturb any existing edges, and creates only one new edge, and the new node is never large, since if it becomes the root it will have no siblings in the general view, and if it does not become the root then it will have no children in the general view.  (In the binary view this corresponds to having no right child or no left child.)

There are only two nodes whose node potentials change, the new node and the old root.  If the old root has a larger key value than the new node, then the new node becomes the root.  They both become mixed nodes with a potential of 400.  If the new node is bigger than the old root, then the old root still has a potential of 400, and the new node becomes a mixed node, because it has no children in the general view, and thus also has a potential of 400.  Thus, if $N$ does not change, the amortized cost is constant.

However, $N$ could increase to the next power of two.  If it does, some mixed nodes may become small and some large nodes may become mixed or even small.  These are decreases, so we can ignore them.  Also, all mixed nodes that remain mixed will have their potential decrease.  What we have to worry about is the edge potential.  However, since there are only $O(n/\log n)$ large nodes, the total edge potential can only increase by an amount that is $o(n)$.  Meanwhile, if $N$ increases, its new value is the new heap size $n$, while the old value was $n/2$, so we release $450n$ units of size potential.  The increase in potential if $N$ doubles is thus $O(n/\log n) - \Theta(n)$, which is negative for large enough $n$.  How large must $n$ be for this to hold?  There were previously at most $2n\/\lg(n/2)$ large nodes, and thus at most that many edges with negative edge potential.  Thus, we need ${7\cdot2n\/\lg(n/2)} < 450n$.  Dividing both sides by $n$, we obtain ${14\/\lg(n/2)} < 450$, or ${7\/\lg(n/2)} < 225$.  Reshuffling, we get
$7 < 225\lg(n/2) = 225(\lg n - 1) = 225\lg n - 225$, or equivalently, $232 < 225\lg n$.  Thus, the asymptotic statement actually holds for all $n>3$.  Since a heap that small contains no large nodes at all, the statement holds unconditionally.
\end{proof}

The following observation will be useful when we analyze decrease-key.
\begin{lemma}  A node's potential is monotone nondecreasing in the size of both of its subtrees (in the binary view) if $n$ is fixed.
\end{lemma}
\begin{proof} We must show that increasing the size of the left or right sub-tree of a node never causes its potential to drop.  This follows immediately from several simple observations.  As long as a node is small, its potential is identically zero and thus monotone.  Since the potential is always non-negative, transitioning from small to non-small can only increase the potential.  Observe that the formulas for mixed nodes and large nodes can be combined, as follows:
\newcommand\both{(\xl,\xr)}
$$400+100\min\left(1,\frac{\min\both}{\lg N}\right)\times\lg(\xl+\xr+1).$$
This formula is monotone in $\xl$ and $\xr$, by inspection.
\end{proof}

\begin{lemma} Decrease-key in a pairing heap takes $O(\log n)$ amortized time.\end{lemma}
\begin{proof} The actual time is constant, so it suffices to bound the potential change.  There are three types of nodes that can change potential: the old root, the decreased node, and the decreased node's ancestors in the binary representation.  The ancestors' potential can only go down, since their subtree size is now smaller and potential is a monotone function of subtree size.  This leaves just two nodes that might change potential: the root and the node that had its key decreased.  But the potential of a node is between zero and $400+100\lg n$ at all times.  This leaves the edge potential and the size potential.  The size potential does not change.  Only $O(1)$ edges were created and destroyed, so the edge potential change due to directly affected edges is negligible.  However, it is possible that there is a large indirect effect: some large ancestors of the decreased node might transition from large to mixed, and if those nodes had an incident edge with negative potential, its potential is now zero.  Fortunately for us, at most $\lg n+O(1)$ edges can undergo this transition.  To see this, let $x$ be the decreased node.  The parent of $x$ may transition from large to mixed as a result of losing $x$, but it can't transition from large to small, because losing $x$ can only affect the size of one of its subtrees, not both.  Likewise, $x$'s grandparent may transition from large to mixed, as well as    $x$'s great-grandparent, and $x$'s great-great-grandparent = great$^2$-grandparent, and so on.  However, $x$'s great$^{\lg n}$-grandparent still has more than $\lg n$ descendants in both subtrees despite losing $x$, and so will not undergo this transition.  Thus, the change in edge potential is $O(\log n)$.
\end{proof}

\begin{lemma} Delete-min in a pairing heap takes $O(\log n)$ amortized time.\end{lemma}
\begin{proof}  We break the analysis into two parts.  Delete-min changes $n$, which means it may change  $N$, which may affect the size potential, and the node potential, and the edge potential.  The first part of our analysis bounds the change in potential due to changing $N$, and the second part deals with the delete-min and associated pairings.

Changing $N$ may change small nodes into mixed or even large, and likewise it may change mixed nodes into large.  In the case of the edge potentials, this works in our favor, since the edge potentials can only go down when then the number of large nodes increase.

If the new value of $N$ is $n$ and the old value was $2n$, then we release $900n$ units of size potential, while in \autoref{linear} we showed that the sum of the mixed and large potentials is at most $750n+o(n)$.  In fact, we can redo the calculation of that lemma slightly more precisely now, by taking advantage of the fact that we now have $N=n$.  There are at most $\frac{2n}{\lg n}$ large nodes, each with a potential of at most $400+100\lg n$, so the total potential of all large nodes is $200n+\frac{800n}{\lg n}$.  We then add the mixed nodes for a total of $750n+\frac{800n}{\lg n}$.  Thus, we release enough size potential if
$900-750=150>\frac{800}{\lg n}$.  If the heap contains at least 64 items, this inequality holds, since $150\cdot6=900>800$.  Could it be that for small heaps, we do not release enough size potential to pay for changing $N$?  If the heap has size 8 or less, then there are no large nodes, and we are then also guaranteed to release enough size potential.  Since the size of the heap must be a power of two when $N$ changes, this leaves the heap sizes of 16 and 32 in question.  A heap of size 16 has at most one large node, so the relevant inequality is
$900\cdot16 > 550\cdot16 + 400+100\lg16$.  Reshuffling: $350\cdot16>400+100\cdot4=800$, which clearly holds.  Finally, a heap of size 32 has at most three large nodes, so we need
$350\cdot32 > 3(400+100\lg32)$.  Dividing both sides by 100, we get
$3.5\cdot32 = 7\cdot16 = 7\cdot8\cdot2 = 56\cdot2 = 112 > 3(4+\lg32)=12+3\cdot5=12+15=27$, which also clearly holds.  It remains to analyze the cost of the pairings performed.

If the root has $c>0$ children, then a delete-min performs $c-1$ pairings and thus takes $c$ units of actual time.  (If the root has $c=0$ children, we are doing delete-min on a heap of size 1, which is trivial.)  The loss of the root causes a potential drop of 400.  Notice first that when two nodes are paired, this does not affect the subtree sizes of any other nodes.  There are several cases to consider, depending on the sizes of nodes that get paired, and also depending on whether it is the first or second pairing pass.  To avoid confusion: whenever we use the notation $|x|$ to refer to the size of a node $x$, if $x$ changed size as a result of the pairing we are analyzing, we mean its initial size.  Finally, when we pair two nodes, the node that becomes the parent of the other is said to have won the pairing, while the other is said to have lost~it.

Let $k$ be the number of pairings performed in the first pass.  The number of pairings performed in the second pass is either $k$ or $k-1$.  We will show that the second pass increases the potential by $O(\log n)$ and that the first pass increases the potential by $O(\log n)-2k$, and thus the amortized cost of delete-min is $O(\log n)$.

We first establish some vocabulary we will use throughout the analysis.  Every pairing performed during the delete-min will be between two adjacent siblings (in the general view) $x$ and $y$, where $x$ is left of $y$; see \autoref{fig:pair}.  (In the binary view, $y$ is the right child of $x$.)  We use $x_L$ to denote $x$'s left subtree (in the binary view), $y_L$ for $y$'s left subtree, and $y_R$ for $y$'s right subtree (which is the subtree containing the siblings right of $y$ in the general view).

\begin{figure} \centering \includegraphics{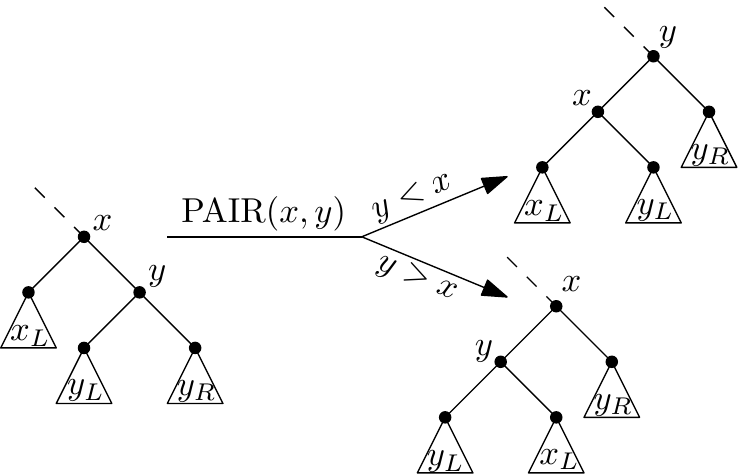}
\caption{A step of the first pairing pass in the binary view.}\label{fig:pair}
\end{figure}

Finally, let $\k_ML$ denote the number of pairings performed in the first pass of a delete-min that involve one mixed node and one large node, and let $\p_ML$ denote the increase in potential as a result of those pairings.  We also define $\k_LL$, $\k_MM$ and so on, analogously.

\paragraph{Large-large.} We start with a large-large pairing.  In this case, the potential function is the same as the classic potential of~\cite{pair}, and the analysis is nearly identical.

\subparagraph{First pass.}  We show that $\p_LL \< -393\k_LL + O(\log n)$.  We will use the fact that
$ab\<\frac14(a+b)^2$.  (Proof:
$ab \< \frac14(a+b)^2  \iff  4ab \< (a+b)^2 = a^2 + b^2 + 2ab  \iff  0 \< a^2+b^2-2ab = (a-b)^2$, and the square of a real number is never negative.)

We know that $\yr>0$, for otherwise $y$ could not be large (see \autoref{fig:pair}).  The initial potential of $x$ is $400+100\lg|x| = 400+100\lg(\xl+2+\yl+\yr)$, and the initial potential of $y$ is $400+100\lg|y| = 400+100\lg(\yl+1+\yr)$.  The potential of $x$ and $y$ after the pairing depends on which of them won the pairing, but the sum of their potentials is the same in either case: $800+100\lg(\xl+\yl+1) + 100\lg(\xl+\yl+2+\yr)$.  The change $P$ in node potential is
\begin{align*}
P&= 100\lg(\xl+\yl+1) - 100\lg(\yl+\yr+1)
\nl< 100\lg(\xl+\yl+1) - 100\lg\yr
\nl= 100[\lg(\xl+\yl+1) +\lg\yr] + 100[-\lg\yr - \lg\yr]
\nl= 100\lg[(\xl+\yl+1)\yr] - 200\lg\yr
\nl\<100\lg\frac14(\xl+\yl+1+\yr)^2 - 200\lg\yr
\nl< 100\lg\frac14|x|^2 - 200\lg\yr
\nl= 200\lg\frac14|x| - 200\lg\yr
\nl= 200(\lg|x|-\lg4) - 200\lg\yr
\nl= 200\lg|x|-200\lg4 - 200\lg\yr
\nl= 200\lg|x| - 400 - 200\lg\yr.
\end{align*}

We now sum the node potential change over all large-large pairings done in the first pass.  Denote the nodes linked by large-large pairings during this pass as $x_1, \dots, x_{2k}$, with $x_i$ being left of $x_{i+1}$.  As a notational convenience, let $L_i=200\lg|x_i|$.  Also, let $x'_i$ denote the right subtree of $x_i$.  Note that for odd $i$, we have $x'_i = x_{i+1}$, but for even $i$, we don't, since there is no guarantee that all large nodes are adjacent to each other.  If we also define $L'_i=200\lg|x'_i|$, then by the calculation above, the $i$th pairing raises the potential by at most $L_{2i-1}-400-L'_{2i}$.  If all large pairings done in the first pass were adjacent, we'd have $L'_{2i}=L_{2i+1}$.  Since they need not be, we have $L'_{2i}\>L_{2i+1}$.  Thus, we have a telescoping sum:
\begin{align*}P&\textstyle
\<\sum_{i=1}^k(L_{2i-1}-400-L'_{2i})\nl
\<\sum_{i=1}^k(L_{2i-1}-400-L_{2i+1})\nl
= -400k+\sum_{i=1}^k(L_{2i-1}-L_{2i+1})\nl
= -400k+\sum_{i=1}^k L_{2i-1}-\sum_{i=1}^k L_{2i+1}\nl
= -400k+(L_1+\sum_{i=2}^k L_{2i-1})-(L_{2k+1}+\sum_{i=1}^{k-1} L_{2i+1})\nl
= -400k+L_1-L_{2k+1}+\sum_{i=2}^k L_{2i-1}-\sum_{i=1}^{k-1} L_{2i+1}\nl
= -400k+L_1-L_{2k+1}+\sum_{i=1}^{k-1} L_{2i+1}-\sum_{i=1}^{k-1} L_{2i+1}\nl
= -400k+L_1-L_{2k+1}\nl
\<-400k+L_1\nl
\<-400k+200\lg n.
\end{align*}

We should also consider the effects of edge potential.  In the course of a pairing, five edges are destroyed and five new edges are created.  Out of these five, three connect a node to its right child (in the binary view) and thus may have non-zero potential.  In the case of a large-large pairing, if the edge that originally connected  $x$ to its parent $p$ had negative potential, then the edge that connects the winner of the pairing to $p$ also has negative potential.  Likewise, if the edge that connects $y$ to $y_R$ had negative potential, then the edge that connects the winner to $y_R$ does too.  However, the edge that connected $x$ to $y$ had negative potential, while the edge that connects the loser to its new right child might not.

Thus, each pairing releases at least 400 units of node potential and costs at most 7 units of edge potential, and therefore $\p_LL \< -393\k_LL + O(\log n)$.

\subparagraph{Second pass.} The second pairing pass repeatedly pairs the two rightmost nodes.  Therefore, one of them has no right siblings in the general representation, which means in the binary representation its right subtree has size zero, which implies that it is not a large node.  Hence, in the second pairing pass there are no large-large pairings.

\paragraph{Mixed-mixed.}  Since $x$ and $y$ are initially mixed, any incident edge has potential zero, so any edge potential change would be a decrease, and thus we can afford to neglect the edge potentials.

\subparagraph{First pass.} We will show that $\p_MM \< -150\k_MM + O(\log n)$.  There are two cases: (1) $x$ and $y$ both have small left subtrees in the binary view, or (2) either $x$ or $y$ has a large left subtree in the binary view.  We first handle the second case: at least one of the two nodes being paired has a large left subtree and hence a small right subtree.  The left-heavy node can not be $x$, because $y$ is contained in the right subtree of $x$, and $y$ can not be mixed if the right subtree of $x$ is small.  Thus we conclude that $y$'s right subtree is small.  This can happen, but it can only happen once during the first pass of a delete-min, because in that case all right siblings of $y$ in the general view will be small.  An arbitrary pairing costs only $O(\log n)$ potential, so this case can not contribute more than $O(\log n)$ to the cost of the first pass.

We now turn to case (1), where $x$ and $y$ both have small left subtrees.  The initial potential of $x$ is $400+100\frac{|x_L|}{\lg N}\lg|x|$, and the initial potential of $y$ is $400+100\frac{|y_L|}{\lg N}\lg|y|$.  Observe that whichever node loses the pairing will have left and right subtrees with sizes $\xl$ and $\yl$.  Thus, both its subtrees will be small, and so the loser becomes a small node with a potential of zero.  There are two sub-cases to consider: (a) the winning node remains mixed, or (b) it becomes large.  We first consider case (a).  Since the winner remains mixed, its new potential is $400+100\frac{\xl+\yl+1}{\lg N}\lg|x|$.  We will make use of the fact that $\yl+1<\lg N$, since the winner is mixed.  The increase $P$ in potential is:
\begin{align*}P&\textstyle
= 400+100\frac{\xl+\yl+1}{\lg N}\lg|x|-(400+100\frac\yl{\lg N}\lg|y|)-(400+100\frac\xl{\lg N}\lg|x|)\nl
= -400 + 100\frac{\xl+\yl+1}{\lg N}\lg|x| - 100\frac\yl{\lg N}\lg|y| - 100\frac\xl{\lg N}\lg|x|\nl
= -400 + 100\frac{\yl+1}{\lg N}\lg|x| - 100\frac\yl{\lg N}\lg|y|\nl
= -400 + 100\frac{\yl+1}{\lg N}\lg(|x_L|+1+|y|) - 100\frac\yl{\lg N}\lg|y|\nl
\<-400 + 100\frac{\yl+1}{\lg N}\lg(\lg N+1+|y|) - 100\frac\yl{\lg N}\lg|y|\nl
\<-400 + 100\frac{\yl+1}{\lg N}\lg(|y_R|+|y|) - 100\frac\yl{\lg N}\lg|y|\nl
< -400 + 100\frac{\yl+1}{\lg N}\lg(|y|+|y|) - 100\frac\yl{\lg N}\lg|y|\nl
= -400 + 100\frac{\yl+1}{\lg N}\lg2|y| - 100\frac\yl{\lg N}\lg|y|\nl
= -400 + 100\frac{\yl+1}{\lg N}(1+\lg|y|) - 100\frac\yl{\lg N}\lg|y|\nl
= -400 + 100\frac{\yl+1}{\lg N}+100\frac{\yl+1}{\lg N}\lg|y| - 100\frac\yl{\lg N}\lg|y|\nl
= -400 + 100\frac{\yl+1}{\lg N}+100\frac1{\lg N}\lg|y|\nl
< -400 + 100 + 100\frac1{\lg N}\lg n\nl
= -300 + 100\frac1{\lg N}\lg n\nl
< -300 + 100\frac1{\lg N}\lg2N\nl
= -300 + 100\frac1{\lg N}(\lg N+1)\nl
= -300 + 100(1+\frac1{\lg N})\nl
= -200 + \frac{100}{\lg N}\nl
\<-200 + \frac{100}{\lg4}\nl
= -200 + \frac{100}{2}\nl
= -200 + 50\nl
= -150.
\end{align*}
Thus, if the node that wins the pairing remains mixed, then at least 150 units of potential are released.  We now turn to case (b) where the node that wins the pairing becomes large.  This can only happen if
$\xl+\yl+1>\lg N$.  In that case, the increase in potential is
\begin{align*}P&\textstyle
= 400 + 100\lg|x| - 100(4+\frac\yl{\lg N}\lg|y|) - 100(4+\frac\xl{\lg N}\lg|x|)\nl
= -400 + 100\lg|x| - 100\frac\yl{\lg N}\lg|y| - 100\frac\xl{\lg N}\lg|x|\nl
< -400 + 100\lg|x| - 100\frac\yl{\lg N}\lg|y| - 100\frac\xl{\lg N}\lg|y|\nl
= -400 + 100\lg|x| - 100\frac{\xl+\yl}{\lg N}\lg|y|\nl
\<-400 + 100\lg|x| - 100\lg|y|\nl
= -400 + 100\lg(\xl+1+|y|) - 100\lg|y|\nl
\<-400 + 100\lg2|y| - 100\lg|y|\nl
= -400 + 100(\lg|y|+1) - 100\lg|y|\nl
= -400 + 100\lg|y|+100 - 100\lg|y|\nl
= -300 + 100\lg|y| - 100\lg|y|\nl
= -300.
\end{align*}
and thus at least 300 units of potential are released.

\subparagraph{Second pass.}  In the second pass, we are guaranteed that $\yr=0$.  The initial potential of $y$ is thus 400.  The initial potential of $x$ depends on which of its subtrees is small.  But in fact we know that $x$ is right-heavy, or else $y$ could not be mixed.  Thus the initial potential of $x$ is
$400+100\frac\xl{\lg N}\lg|x|=400+100\frac\xl{\lg N}\lg(\xl+2+\yl)$.  Whichever node wins the pairing will have a final potential of 400 (because it will have no right siblings).  Whichever node loses the pairing will have a final potential of
$400+100\frac\xl{\lg N}\lg(\xl+1+\yl)$.  Thus, there is no potential gain.

\paragraph{Mixed-small and large-small.} These types of pairings can only happen once during the first pass. To see this, observe that, in the general view, all right siblings of a small node are small.  Therefore we have $\k_MS+\k_LS \< 2$.  An arbitrary pairing costs only $O(\log n)$ potential, so $\p_MS+\p_LS$ is $O(\log n)$.  (The winner of such a pairing is no longer small, so this type of pairing can happen only once during the second pass as well.)

\paragraph{Small-small.}  We show that the number $\k_SS$ of small-small pairings performed in both passes is $O(\log n)$, and that the potential increase $\p_SS$ caused by said pairings is also $O(\log n)$.  In one pass, there are fewer than $\lg N$ small-pairings, because a small node has few right siblings.  By the same logic as for mixed-mixed, the loser of a small-small pairing remains small.  The winner may remain small, in which case the pairing costs no potential.  Or the winner may become mixed.  However, that can only happen once per pass.  For if the winner is mixed, that means that the combined subtree sizes of the two nodes exceeded $\lg N$, which means none of the left siblings were small.  Thus, in the first pass, only the first small-small pairing may have the winner become mixed.  Thus, $\p_SS$ is $O(\log n)$.

\paragraph{Mixed-large.}  We show that the heap potential can increase by at most $O(\log  n)$ as a result of all mixed-large pairings performed.  We begin with the second pass this time, to get the easy case out of the way first.  The mixed-large pairings of the second pass cause no potential increase at all.

\subparagraph{Second pass.}  In the second pass, we are guaranteed that $\yr=0$.  The initial potential of  $y$ is thus 400.  The initial potential of $x$ is $400+100\lg|x|=400+100\lg(\xl+2+\yl) $.  Whichever node wins the pairing will have a final potential of 400 (because it will have no right siblings).  Whichever node loses the pairing will have a final potential of $400+100\lg(\xl+1+\yl)$.  Thus, there is no increase in node potential.  At first it seems that the edge potential could increase, for even though $\yl\>\lg N$ (and likewise for $x_L$), there is no guarantee that $y_L$ (or $x_L$) is a large node.  On the other hand, there is also no guarantee that the parent of $x$ in the binary view is not a large node.  Thus, it is possible that we lose an edge with negative potential and have nothing to replace it with.  However, this can only happen once, because the loser of such a pairing becomes a large node and will play the role of $y_L$ in the next pairing, and once any pairing has a large $y_L$, all subsequent ones will too.

\subparagraph{First pass.}  We have three cases to consider, depending on which of $\xl$, $\yl$, or $\yr$ is small.  The case where $\yr\<\lg N$ is easy to dispense with, as it can only happen once during a delete-min.  Observe that in the other two cases, the edge potential can't increase, because the winner of the pairing is large.  If $\yl\<\lg N$, the initial potential of $x$ is $400+100\lg|x|$, and the initial potential of $y$ is $400+100\frac\yl{\lg N}\lg|y|$.  After the pairing, the loser is a mixed node with potential $400 + 100\frac\yl{\lg N}\lg(\xl+1+\yl)$.  The winner is large with a potential of $400+100\lg|x|$.  Thus, the increase $P$ in potential is
\begin{align*}P&\textstyle
= 100\frac\yl{\lg N}\lg(\xl+1+\yl) - 100\frac\yl{\lg N}\lg|y|\nl
= 100\frac\yl{\lg N}[\lg(\xl+1+\yl) - \lg|y|]\nl
= 100\frac\yl{\lg N}[\lg(\xl+1+\yl) - \lg(\yl+1+\yr)]\nl
\<100[\lg(\xl+1+\yl) - \lg(\yl+1+\yr)]\nl
< 100[\lg|x| - \lg(\yl+1+\yr)]\nl
< 100[\lg|x| - \lg\yr].
\end{align*}
Observe that when we sum over all such mixed-large pairings, we get a telescoping sum of the same form as the one that arose in the analysis of large-large pairings, and thus the combined potential increase for all such pairings is $O(\log n)$.  We now turn to the last case, where $\xl\<\lg N$.  The initial potential of $x$ is now $400+100\frac\xl{\lg N}\lg|x|$, and that of $y$ is $400+100\lg|y|$.  The potential of the winner of the pairing is $400+100\lg|x|$, and the potential of the loser is
$400+100\frac\xl{\lg N}\lg(\xl+\yl+1)$, so the increase $P$ in potential is
\begin{align*}P&\textstyle
= 100\lg|x| + 100\frac\xl{\lg N}\lg(\xl+\yl+1) - 100\lg|y| - 100\frac\xl{\lg N}\lg|x|\nl
< 100\lg|x| - 100\lg|y|\nl
< 100\lg|x| - 100\lg\yr.
\end{align*}
Summing over all such mixed-large pairings, this sum again telescopes in the same way as in the large-large case.  Thus, all mixed-large pairings combined cost only $O(\log n)$ potential.  However, unlike small-small, the actual cost can be far greater, and unlike large-large, we may not release enough node potential to pay for it.  What saves us is the edge potentials.

Call a mixed-large pairing \emph{normal} if $\yr>\lg N$.  Observe that during the first pass, at most one mixed-large pairing can be abnormal, because if $\yr\<\lg N$, all its right siblings (in the general view) are small nodes.  We show that three consecutive normal mixed-large pairings release at least 7 units of potential: enough to pay for those three pairings, with 4 units left over.  First, observe that the winner of a normal mixed-large pairing must be large.  Thus, for every three consecutive normal mixed-large pairings, at least two large siblings become adjacent that were not adjacent before, and thus some edge has its potential go from 0 to $-7$.  (We must of course be careful that this is not offset by some other edge nearby undergoing the opposite transition, but indeed, we are safe here, because the winner of the pairing is large.)  Thus, even though we don't release enough node potential to pay for all mixed large pairings, there is hope that if there are so many of them that they tend to be consecutive, the edge potentials can pay for them instead, while if there are not so many that they tend to be consecutive, perhaps they do not dominate the cost of the first pass.  We will soon see that this is indeed the case.

\paragraph{Bringing it all together.}  We can now calculate the total amortized runtime of delete-min.  The actual work done in the second pass is the same as that of the first pass, and the second pass causes at most a logarithmic increase in potential.  Thus, we must show that $\Delta\phi \< -2k+O(\log n)$, where $k$ is the number of pairings done by the first pass and $\Delta\phi$ is the change in potential due to the first pass.  We have $k=
\k_LL+\k_MM+\k_SS+\k_ML+\k_MS+\k_LS\<
\k_LL+\k_MM+\lg N+\k_ML+1=
\k_LL+\k_MM+\k_ML+O(\log n)$.
The increase in potential from the first pass is
$\Delta\phi=\p_LL+\p_MM+\p_SS+\p_ML+\p_MS+\p_LS \< O(\log n) - 393\k_LL - 150\k_MM$.  Summing the actual work with the node potential change, we obtain
$O(\log n) - 392\k_LL - 149\k_MM + \k_ML \< O(\log n) - 149(\k_LL + \k_MM) + \k_ML$.  Thus, the question is whether most of those terms cancel, leaving us with a $O(\log n)-k$ amortized cost.  There are two cases to consider, depending on how large $\k_ML$ is.  If at most $74\/75$ of all pairings done are mixed-large (that is, $\k_ML<\frac{74}{75} k$, or rather $\k_ML<\frac{74}{75}(\k_LL+\k_MM+\k_ML)$, or equivalently
$\k_ML < 74\k_LL + 74\k_MM$), then the amortized cost $C$ of the first pass is at most
\begin{align*}
C&\<O(\log n) - 149(\k_LL + \k_MM) + \k_ML\nl
\<O(\log n) - 149(\k_LL + \k_MM) + 74(\k_LL + \k_MM)\nl
= O(\log n) - 75(\k_LL + \k_MM)\nl
= O(\log n) - (74+1)(\k_LL + \k_MM)\nl
= O(\log n) - 74(\k_LL + \k_MM)+\k_LL + \k_MM\nl
\<O(\log n) - \k_ML + \k_LL + \k_MM\nl
\<O(\log n) - k.
\end{align*}
Since the second pass only increases the potential by $O(\log n)$ and its actual cost is $k$, the cost for the whole delete-min is $O(\log n)-k+k=O(\log n)$.

That leaves the case where more than $\frac{74}{75}$ of parings are mixed-large ones.  In fact, we will use the weaker assumption that at least $20\/21$ of the pairings in the first pass are mixed-large.  We divide the pairings into groups of three: the first three pairings, the second three, and so on.  If a group consists of only mixed-large pairings, then it releases 7 units of potential, for a total amortized cost of $-4$.  There are $k/3$ groups (give or take divisibility by three), and all but $k/21$ of those groups consist entirely of mixed-large pairings.  These $k/3 - k/21 = 6k/21 = 2k/7$ groups release $8k/7$ units of spare potential.  The remaining $k/21$ groups require $k/7$ units of potential to pay for them, leaving $k$ units of spare potential, which we use to pay for the second pairing pass.
\end{proof}

\section{Final words}
It would be interesting to extend this potential function to so that merging two pairing heaps is fast.  The amortized time using this potential function is, in the worst case, linear, because the two heaps might have sizes that are adjacent powers of 2 (e.g., 1024 and 2048), and thus the size potential of the new heap is linear in the combined size of the old heap, whereas the size potentials of the old heaps are both zero.

It would also be interesting to see whether a similar potential function can be made to work for splay trees; the one presented here does not.  In particular, a splay where every double-rotation is a zig-zag does not release enough potential if the node $x$ being accessed is large and all nodes on the path to $x$ are mixed, so the amortized cost of the splay would be super-logarithmic.

\end{document}